\documentclass[onecolumn,journal,final]{IEEEtran}

\usepackage{tikz}
\usetikzlibrary{}
\usepgfmodule{oo}

\usepackage[]{hyperref}
\usepackage{amsthm}
\usepackage{amsmath}
\usepackage{amssymb}
\usepackage{quantum}
\usepackage{stmaryrd}
\usepackage{graphicx}
\usepackage[percent]{overpic}
\allowdisplaybreaks[4]

\usepackage{pifont}
\newcommand{\cmark}{\ding{51}}%
\newcommand{\xmark}{\ding{55}}%

\begin{document}

\title{R\'enyi Bounds on Information Combining}

 \author{%
   \IEEEauthorblockN{Christoph Hirche}\\
   \IEEEauthorblockA{QMATH, Department of Mathematical Sciences, University of Copenhagen\\
                     Universitetsparken 5, 2100 Copenhagen, Denmark.\\
                     Email: christoph.hirche@gmail.com}
 }

\maketitle

\begin{abstract}
Bounds on information combining are entropic inequalities that determine how the information, or entropy, of a set of random variables can change when they are combined in certain prescribed ways. Such bounds play an important role in information theory, particularly in coding and Shannon theory. The arguably most elementary kind of information combining is the addition of two binary random variables, i.e. a CNOT gate, and the resulting quantities are fundamental when investigating belief propagation and polar coding.

In this work we will generalize the concept to R\'enyi entropies. We give optimal bounds on the conditional R\'enyi entropy after combination, based on a certain convexity or concavity property and discuss when this property indeed holds. Since there is no generally agreed upon definition of the conditional R\'enyi entropy, we consider four different versions from the literature. 

Finally, we discuss the application of these bounds to the polarization of R\'enyi entropies under polar codes. 
\end{abstract}
%

\section{Introduction}
Many tasks in information theory are concerned with the evolution of random variables and their corresponding entropies under certain ``combining operations''. 
A particularly relevant example of such an operation is the addition of two independent random variables (with values in some group). In this case, the entropy can be easily computed since we know that the addition of two random variables has a probability distribution which corresponds to the convolution of the probability distributions of the individual random variables. 
The picture changes when we have random variables with \emph{side information}. Now we are interested in the entropy of the sum conditioned on all the available side information. Evaluating this is substantially more difficult. The field of \textit{bounds on information combining} is concerned with finding optimal entropic bounds on the resulting conditional entropy. 

A particularly important case is that of two binary random variables with side information of arbitrary finite dimension. This is the setting we will consider in this work as it has many applications e.g. in coding theory for polar codes~\cite{AT14} and Shannon theory~\cite{GKBook, RU08}. 
An optimal lower bound for the resulting Shannon entropy was given by Wyner and Ziv in \cite{WZ73}, the well known \textit{Mrs. Gerber's Lemma}. 
Following this result, additional approaches to the problem have been found which also led to an \emph{upper} bound on the conditional entropy of the combined random variables. One proof method and several additional applications can be found e.g. in \cite{RU08} along with the optimal upper bound.

In this work, we discuss the extension of these information combining bounds to R\'enyi entropies. While the unconditional R\'enyi entropy is clearly defined, one can find many different definitions of the conditional R\'enyi entropy in the literature. We focus on four  of the most commonly found conditional  R\'enyi entropies which we detail in Section~\ref{sec:Renyi}.  

In all cases, we find the optimal lower and upper bounds given a certain convexity or concavity property of an associated function. In many cases we will also give the ranges of $\alpha$ for which the desired property holds. Furthermore, we identify some remarkable cases where the bounds hold with equality. These results are briefly summarized in Figure~\ref{Fig:Results}. 
As in the Shannon entropy case, our bounds are optimal in the sense that they are attained by the binary symmetric channel and the binary erasure channel with equality. 

Finally, we briefly discuss the application of our results to channel polarization. We show that our results can easily be combined with a simplified proof strategy for polarization from~\cite{AT14} to show polarization of $H_\alpha^J$ under the polar coding transformation, therefore reproducing a result previously found in~\cite{zheng2019polarization}.

\section{Background and previous results}

In this section we will give a brief overview over the known results for the Shannon entropy case. We consider two pairs of random variables, $(X_1,Y_1)$ and $(X_2,Y_2)$, with $X_i$ binary and $Y_i$ being finite dimensional side information. The main quantities under investigation are 
\begin{equation}\label{mainterm}
H(X_1+X_2|Y_1Y_2)
\end{equation}
and 
\begin{equation}\label{eq:varo-minus}
H(X_2 | X_1+X_2,Y_1Y_2),
\end{equation}
which are related by
\begin{equation}\label{additiveentropies}
H(X_1+X_2|Y_1Y_2) + H(X_2 | X_1+X_2,Y_1Y_2) = H(X_1|Y_1) + H(X_2|Y_2).
\end{equation}

The optimal bounds are well known to be as follows~\cite{WZ73, RU08}:
\begin{align}
h\left(h^{-1}(H_1)\ast h^{-1}(H_2)\right) &\leq H(X_1+X_2|Y_1Y_2) \leq \log 2 - \frac{(\log 2 - H_1)(\log 2 - H_2)}{\log 2}, \label{minus-bounds}
\end{align}
where the lower bound is called the Mrs. Gerber's Lemma, and
\begin{align}
\frac{H_1 H_2}{\log 2} &\leq H(X_2 | X_1+X_2,Y_1Y_2) \leq H_1 + H_2 - h\left(h^{-1}(H_1)\ast h^{-1}(H_2)\right),  \label{plus-bounds}
\end{align}
with $H_1 = H(X_1|Y_1)$ and $H_2 = H(X_2|Y_2)$, $h(\cdot)$ the binary entropy, $h^{-1}(\cdot)$ its inverse and $a\ast b:=a(1-b) + (1-a)b$ the binary convolution.
In many situations, it is more intuitive to look at the special case where the two underlying entropies are equal $\left( H = H_1 = H_2\right)$. In that case we can state the following inequalities
\begin{align}
0.799\, \frac{H(\log2-H)}{\log2} + H &\leq h\left(h^{-1}(H_1)\ast h^{-1}(H_2)\right) \label{HHlowerbound}\\
&\leq H(X_1+X_2|Y_1Y_2) \leq  \frac{H(\log2-H)}{\log2} + H, \label{HHupperbound}
\end{align}
where the first is an additional convenient lower bound from~\cite{GX15} and the other two follow from Equation~\ref{minus-bounds}.

It follows from Equation~\eqref{additiveentropies} that it is sufficient to prove the inequalities for either Equation~\eqref{minus-bounds} or ~\eqref{plus-bounds}.  As we will see, most R\'enyi generalization of conditional entropy do not obey Equation~\eqref{additiveentropies}, due to the lack of an appropriate chain rule. In the following, we will focus on bounding the R\'enyi  generalizations of Equation~\eqref{mainterm} and leave bounds on those of Equation~\eqref{eq:varo-minus} for future research. 

Finally, we remark that, considering the random variables $X_i$ and $Y_i$ as input and output of a channel, respectively, it is also known for which channels equality is achieved in the above equations (see e.g. \cite{RU08}). For the lower bound in Equation \eqref{minus-bounds} this is the binary symmetric channel (BSC) and for the upper bound it is the binary erasure channel (BEC). For clarity, it will later be useful to refer to the generalized bounds as BSC and BEC bounds according to the channel that achieves equality.

 
 \section{Conditional R\'enyi entropies}\label{sec:Renyi}
 
The R\'enyi entropy $H_\alpha(X)$ is defined as
\begin{align}
H_\alpha(X) := \frac{1}{1-\alpha} \log \sum_x p(x)^\alpha.
\end{align} 
 In analogy to the case of the Shannon entropy, we have
 \begin{align}
 H_\alpha(X_1+X_2) = h_\alpha\left(h_\alpha^{-1}( H_\alpha(X_1))\ast h_\alpha^{-1}( H_\alpha(X_2))\right).
 \end{align} 
While this definition is well established, the picture becomes much less clear when we consider the conditional R\'enyi entropy. In fact, a multitude of different definitions exists~\cite{FB14}, each of which has found a number of applications. In the following, we will consider the four most commonly used definitions: 
\begin{align}
H_{\alpha}^{A}(X|Y) &:= \frac{\alpha}{1-\alpha}\log\left(\sum_y p(y) \left(\sum_x p(x|y)^\alpha\right)^{\frac{1}{\alpha}}\right),\\
H_{\alpha}^{H}(X|Y) &:= \frac{1}{1-\alpha}\log\left( \sum_y\sum_x p(y) p(x|y)^\alpha\right), \\
H_\alpha^J(X|Y) &:= \frac{1}{1-\alpha}\left[\log\sum_{x,y} p(x,y)^\alpha  - \log\sum_y p(y)^\alpha \right], \\
H_\alpha^C(X|Y) &:= \frac{1}{1-\alpha} \sum_y p(y) \log \sum_x p(x|y)^\alpha, 
\end{align}
 where the first was originally given by Arimoto~\cite{Arimoto77}, the second by Hayashi and Skoric \etal~\cite{Hayashi11, Skoric11}, the third by Jizba and Arimitsu~\cite{Jizba04} and the fourth by Cachin~\cite{Cachin97}\footnote{Some of these definitions have implicitly appeared in the literature before, see e.g.~\cite{csiszar1995generalized}}.  
 
In the remainder of this section we will discuss some useful properties of the above definitions, see also~\cite{teixeira2012conditional, IS13, FB14}. 
We have,
\begin{align}
\lim_{\alpha\rightarrow 1} H_\alpha(X) = H(X)
\end{align}
and we would expect for every conditional R\'enyi entropy $H^*_\alpha$ that the following holds
\begin{align}
\lim_{\alpha\rightarrow 1} H^*_\alpha(X|Y) = H(X|Y). 
\end{align}
Two further important properties are monotonicity
  \begin{align}
 H^*_\alpha(X|YZ) \leq H^*_\alpha(X|Z)
 \end{align}
and the (strong) chain rule
 \begin{align}
 H^*_\alpha(X|YZ) = H^*_\alpha(XY|Z) - H^*_\alpha(Y|Z),
 \end{align}
 both of which hold for the conditional Shannon entropy, but not necessarily for all R\'enyi generalizations. 
 
However, note that while conditional R\'enyi entropies generally are not even subadditive, they are additive for independent pairs of random variables, i.e. 
 \begin{align}
 H^*_\alpha(X_1X_2 | Y_1Y_2) =  H^*_\alpha(X_1|Y_1) +  H^*_\alpha(X_2|Y_2), 
 \end{align}
 when $(X_1,Y_1)$ and $(X_2,Y_2)$ are independent of each other. 
 
 Finally, a commonly used quantity is the so-called \textit{min-entropy}, defined as 
 \begin{align}
 H_\infty(X) = -\log\max_x p(x)
 \end{align}
 and its conditional version
 \begin{align}
 H_\infty(X|Y) = -\log\sum_y p(y)\max_x p(x|y). 
 \end{align}
 Ideally, we would have
 \begin{align}
 \lim_{\alpha\rightarrow\infty} H^*_\alpha(X|Y) = H_\infty(X|Y). 
 \end{align}
 Which of the above properties hold for the several conditional R\'enyi entropies considered is summarized in Table~\ref{table:properties}. 
 
 \begin{table}[t!]
\renewcommand{\arraystretch}{1.3}
\caption{Table of several properties for the considered R\'enyi entropies.}
\label{table:properties}
\centering
\begin{tabular}{c c c c| c } 
 $\lim_{\alpha\rightarrow 1}$ & $\lim_{\alpha\rightarrow\infty}$ & Chain rule & Monotonicity & * \\ 
 \hline
 \cmark~\cite{Arimoto77} & \cmark~\cite{FB14, IS13} & (weak)~\cite{FB14} & \cmark~\cite{Arimoto77} & $H^A_\alpha$ \\ 
\cmark & \xmark & \xmark~\cite{FB14} & \cmark~\cite{IS13}  & $H^H_\alpha$  \\
\cmark & \xmark & \cmark & \xmark~\cite{FB14} & $H^J_\alpha$  \\
 \cmark & \xmark & \xmark~\cite{Cachin97} & \xmark~\cite{FB14} & $H^C_\alpha$ 
\end{tabular}
\end{table}

 In the next section, we will show how to generalize information combining bounds to the R\'enyi entropies discussed above. We will devote one section to each of the entropies, starting with Arimotos conditional entropy. 
 
 
\section{R\'enyi Bounds on Information Combining}
\subsection{Bounds for $H_{\alpha}^{A}(X_1+X_2|Y_1Y_2)$} 

 A particular feature of the Shannon entropy that allowed to prove the classical bounds, is that a conditional Shannon entropy can be written as convex combination of unconditioned Shannon entropies. This important feature does not hold for R\'enyi entropies in general. Nevertheless, for $H_{\alpha}^{A}$ we have the following equality, providing us with a similar tool: 
 \begin{align} 
 H_{\alpha}^{A}(X|Y) &=\frac{\alpha}{1-\alpha}\log\left(\sum_y p(y) e^{\frac{1-\alpha}{\alpha} H_{\alpha}^{A}(X|Y=y) }\right).
 \end{align}
 This motivates us to define the following quantity: 
 \begin{align}
  K_{\alpha}^{A}(X|Y) &=  e^{\frac{1-\alpha}{\alpha} H_{\alpha}^{A}(X|Y) },
 \end{align}
 which takes values on $[1, \delta_\alpha^A := 2^{\frac{1-\alpha}{\alpha}} ]$. 
 For a more convenient notation we will furthermore use, in analogy to the binary entropy, $k_{\alpha}^{A}(p)$ whenever $X$ is a binary random variable with probability distribution $\{p,1-p\}$ (and trivial conditioning system): $K^A_\alpha(X) = k^A_\alpha(p)$. 
 
  We will see that the crucial quantity now is the following:
 \begin{align}
\kk^A_\alpha(x,y) = {k^A_{\alpha}}\left({k^A_\alpha}^{-1}(x) \ast {k^A_\alpha}^{-1}(y) \right).   \label{kfuncA}
 \end{align}

 Following the proof technique of the Shannon bounds on information combining, we get the following results. 
 \begin{thm}[$H_{\alpha}^{A}$ BSC-bound]\label{thm:BSC-A}
 If, for a given $\alpha$, the function $\kk^A_\alpha(x,y)$  is convex in $x$ for fixed $y$ and vice versa, then the following holds: \\
 If $\alpha >1$, 
 \begin{align}
 H_\alpha^A(X_1+X_2|Y_1Y_2) \leq h_\alpha\left(h_\alpha^{-1}( H_\alpha^A(X_1|Y_1))\ast h_\alpha^{-1}( H_\alpha^A(X_2|Y_2))\right). \label{BSCupA}
 \end{align}
  If $\alpha <1$, 
 \begin{align}
 H_\alpha^A(X_1+X_2|Y_1Y_2) \geq h_\alpha\left(h_\alpha^{-1}( H_\alpha^A(X_1|Y_1))\ast h_\alpha^{-1}( H_\alpha^A(X_2|Y_2))\right). \label{BSCdownA}
 \end{align}
If $\kk^A_\alpha(x,y)$ is concave instead, the inequalities hold with $\leq$ and $\geq$ exchanged. These bounds are optimal, in the sense that equality is achieved by binary symmetric channels. 
 \end{thm}
 \begin{proof}
To prove the theorem, we will show Equation \ref{BSCupA} in the case of $\kk^A_\alpha(x,y)$ being convex. All the other combinations, for \ref{BSCupA} and \ref{BSCdownA}, follow directly. The prove strategy is similar to the proof in the Shannon entropy setting. 

Consider the following chain of equations: 
 \begin{align} 
 H_{\alpha}^{A}(X_1+X_2|Y_1Y_2) &=\frac{\alpha}{1-\alpha}\log\left(\sum_{y_1,y_2} p(y_1)p(y_2) e^{\frac{1-\alpha}{\alpha} H_{\alpha}^{A}(X_1+X_2|Y_1=y_1,Y_2=y_2) }\right) \\
	&=\frac{\alpha}{1-\alpha}\log\left(\sum_{y_1,y_2} p(y_1)p(y_2) e^{\frac{1-\alpha}{\alpha} h_{\alpha}(h_\alpha^{-1}(H_\alpha^A(X_1|Y_1=y_1))\ast h_\alpha^{-1}(H_\alpha^A(X_2|Y_2=y_2))) }\right) \\
	&=\frac{\alpha}{1-\alpha}\log\left(\sum_{y_1,y_2} p(y_1)p(y_2) {k^A_{\alpha}}\left({k^A_\alpha}^{-1}(K_\alpha^A(X_1|Y_1=y_1))\ast {k^A_\alpha}^{-1}(K_\alpha^A(X_2|Y_2=y_2)) \right) \right)	\\
	&\leq\frac{\alpha}{1-\alpha}\log\left({k^A_{\alpha}}\left({k^A_\alpha}^{-1}\left(\sum_{y_1} p(y_1)K_\alpha^A(X_1|Y_1=y_1)\right)\ast {k^A_\alpha}^{-1}\left(\sum_{y_2} p(y_2)K_\alpha^A(X_2|Y_2=y_2)\right) \right) \right)	 \\
	&= h_{\alpha}\left(h_\alpha^{-1}(H_\alpha^A(X_1|Y_1))\ast h_\alpha^{-1}(H_\alpha^A(X_2|Y_2))\right),
\end{align}
where all equalities follow simply by definition and rearranging. The inequality follows by using the convexity of $\kk^A_\alpha(x,y)$ twice, once in the first argument and once in the second. 
\end{proof}

 This result gives an extension of the lower bound for Shannon entropy information combining. 
 In a similar fashion, we can also generalize the upper bound: 

\begin{thm}[$H_{\alpha}^{A}$ BEC-bound]\label{thm:BEC-A}
 If, for a given $\alpha$, the function $\kk^A_\alpha(x,y)$
 is convex in $x$ for fixed $y$ and vice versa, then the following holds: \\
 If $\alpha >1$, 
 \begin{align}
 H_\alpha^A(X_1+X_2|Y_1Y_2) &\geq \frac{\alpha}{1-\alpha}\log\frac{(\delta_\alpha^A - K_\alpha^A(X_1|Y_1))(\delta_\alpha^A - K_\alpha^A(X_2|Y_2))}{1-\delta_\alpha^A} + \delta_\alpha^A \label{BEC11A}.
 \end{align}
  If $\alpha <1$, 
\begin{align}
 H_\alpha^A(X_1+X_2|Y_1Y_2) &\leq \frac{\alpha}{1-\alpha}\log\frac{(\delta_\alpha^A - K_\alpha^A(X_1|Y_1))(\delta_\alpha^A - K_\alpha^A(X_2|Y_2))}{1-\delta_\alpha^A} + \delta_\alpha^A .
 \end{align}
If $\kk^A_\alpha(x,y)$ is concave instead, the inequalities hold with $\leq$ and $\geq$ exchanged. These bounds are optimal, in the sense that equality is achieved by binary erasure channels. 
\end{thm}
\begin{proof}
 We will show Equation \ref{BEC11A} in the case of $\kk^A_\alpha(x,y)$ being convex. All the other statements follow similarly. 
 
 The crucial part of this proof is to find an upper bound on $\kk^A_\alpha(x,y)$. Since we assume that $\kk^A_\alpha(x,y)$ is convex, an upper bound is given by the line connecting the endpoints of its graph. Let's for now consider $y$ to be fixed. It can easily be seen that $\kk^A_\alpha(\delta^A_\alpha,y)=\delta^A_\alpha$ and $\kk^A_\alpha(1,y)=y$. Therefore we can give the bound
 \begin{align}
 \kk^A_\alpha(x,y) &\leq \frac{y-\delta^A_\alpha}{1-\delta^A_\alpha} x + \frac{\delta^A_\alpha(1-y)}{1-\delta^A_\alpha} \\
 &= \frac{(\delta^A_\alpha-x)(\delta^A_\alpha-y)}{1-\delta^A_\alpha} + \delta^A_\alpha
 \end{align}
 and the same for $x$ fixed. 
 With this, the proof follows via the same arguments as for the previous theorem: 
  \begin{align} 
 H_{\alpha}^{A}(X_1+X_2|Y_1Y_2) &=\frac{\alpha}{1-\alpha}\log\left(\sum_{y_1,y_2} p(y_1)p(y_2) e^{\frac{1-\alpha}{\alpha} H_{\alpha}^{A}(X_1+X_2|Y_1=y_1,Y_2=y_2) }\right) \\
	&=\frac{\alpha}{1-\alpha}\log\left(\sum_{y_1,y_2} p(y_1)p(y_2) e^{\frac{1-\alpha}{\alpha} h_{\alpha}(h_\alpha^{-1}(H_\alpha^A(X_1|Y_1=y_1))\ast h_\alpha^{-1}(H_\alpha^A(X_2|Y_2=y_2))) }\right) \\
	&=\frac{\alpha}{1-\alpha}\log\left(\sum_{y_1,y_2} p(y_1)p(y_2) {k^A_{\alpha}}\left({k^A_\alpha}^{-1}(K_\alpha^A(X_1|Y_1=y_1))\ast {k^A_\alpha}^{-1}(K_\alpha^A(X_2|Y_2=y_2)) \right) \right)	\\
	&\leq \frac{\alpha}{1-\alpha}\log\left(\sum_{y_1,y_2} p(y_1)p(y_2)  \frac{(\delta^A_\alpha-K_\alpha^A(X_1|Y_1=y_1))(\delta^A_\alpha-K_\alpha^A(X_2|Y_2=y_2))}{1-\delta^A_\alpha} + \delta^A_\alpha   \right) \\
	&= \frac{\alpha}{1-\alpha}\log\left( \frac{(\delta^A_\alpha-K_\alpha^A(X_1|Y_1))(\delta^A_\alpha-K_\alpha^A(X_2|Y_2))}{1-\delta^A_\alpha} + \delta^A_\alpha   \right)
\end{align}
 \end{proof} 
 
 Now, the crucial question is: When does the function $\kk^A_\alpha(x,y)$ posses the desired convexity or concavity property? 
 
 One important result in this direction can be found in \cite{HASC18}, which handles the different, but related, problem of generalizing information bottleneck functions. We rephrase it here for our convenience.
 \begin{lem}[Convexity result from \cite{HASC18}]
 For $\alpha\geq 2$, the function
  \begin{align}
 k_{\alpha}^A({k_\alpha^A}^{-1}(x)\ast {k_\alpha^A}^{-1}(y))
 \end{align}
 is convex in $x$ for fixed $y$ and vice versa. 
 \end{lem}
 
The natural question is now: Is the function always convex? It turns out, that this is not the case. In the following lemma we show by concrete examples that the function is neither convex nor concave for some values of $\alpha$. 
 \begin{lem}[Counterexample for $\kk^A_\alpha(x,y)$]\label{CE-A}
 For $\alpha\in(1.58, 1.97)$, the function
  \begin{align}
 k_{\alpha}^A({k_\alpha^A}^{-1}(x)\ast {k_\alpha^A}^{-1}(y))
 \end{align}
 is neither convex nor concave in $x$ for fixed $y$ and vice versa. 
 \end{lem} 
 \begin{proof}
 Remember that the convexity or concavity of $\kk^A_\alpha(x,y)$ with some $\alpha$ would directly imply the validity of the BSC and BEC bounds for all channels and that value of $\alpha$. This means also that there would be a natural order between the two, since one is an upper bound and one a lower bound, again for all channels. In the following, we will show that for some values of $\alpha$ this order relation does not exist, i.e. for some channels the order is opposite to what it is for other channels. Therefore, for these values of $\alpha$ the function could neither have been convex nor concave. 
 
 To this end, let's consider both bounds for a binary symmetric channel with crossover probability $p$, denoted $BSC(p)$. 
 We have 
 \begin{align}
 H_{\alpha}^{A}(BSC(p)) = h_\alpha(p)
 \end{align}
 and therefore the BSC bound for two of these channels is given by $h_\alpha(p\ast p)$. Evaluating the BEC bound for two $BSC(p)$ and defining the difference between both as
 \begin{align}
 \Delta^A(BSC(p),\alpha) := h(p\ast p)  -  \frac{\alpha}{1-\alpha}\log\frac{(\delta_\alpha^A - K_\alpha^A(BSC(p)))^2}{1-\delta_\alpha^A} + \delta_\alpha^A
 \end{align}
 gives us an easily testable quantity. 
 We simply verify numerically that 
 \begin{align}
 \Delta^A(BSC(10^{-6}), \alpha) < 0 \quad\text{for}\; \alpha\in(1,1.97), \\
 \Delta^A(BSC(0.49), \alpha) > 0 \quad\text{for}\; \alpha\in(1.58,2).
 \end{align}
 By the above argument this leads to the desired result. The behavior is illustrated in Figure~\ref{Fig:CE-A} for similar values of $\alpha$. 
 \end{proof}
 
\begin{figure}[t!]
\centering
\begin{overpic}[scale=0.4]{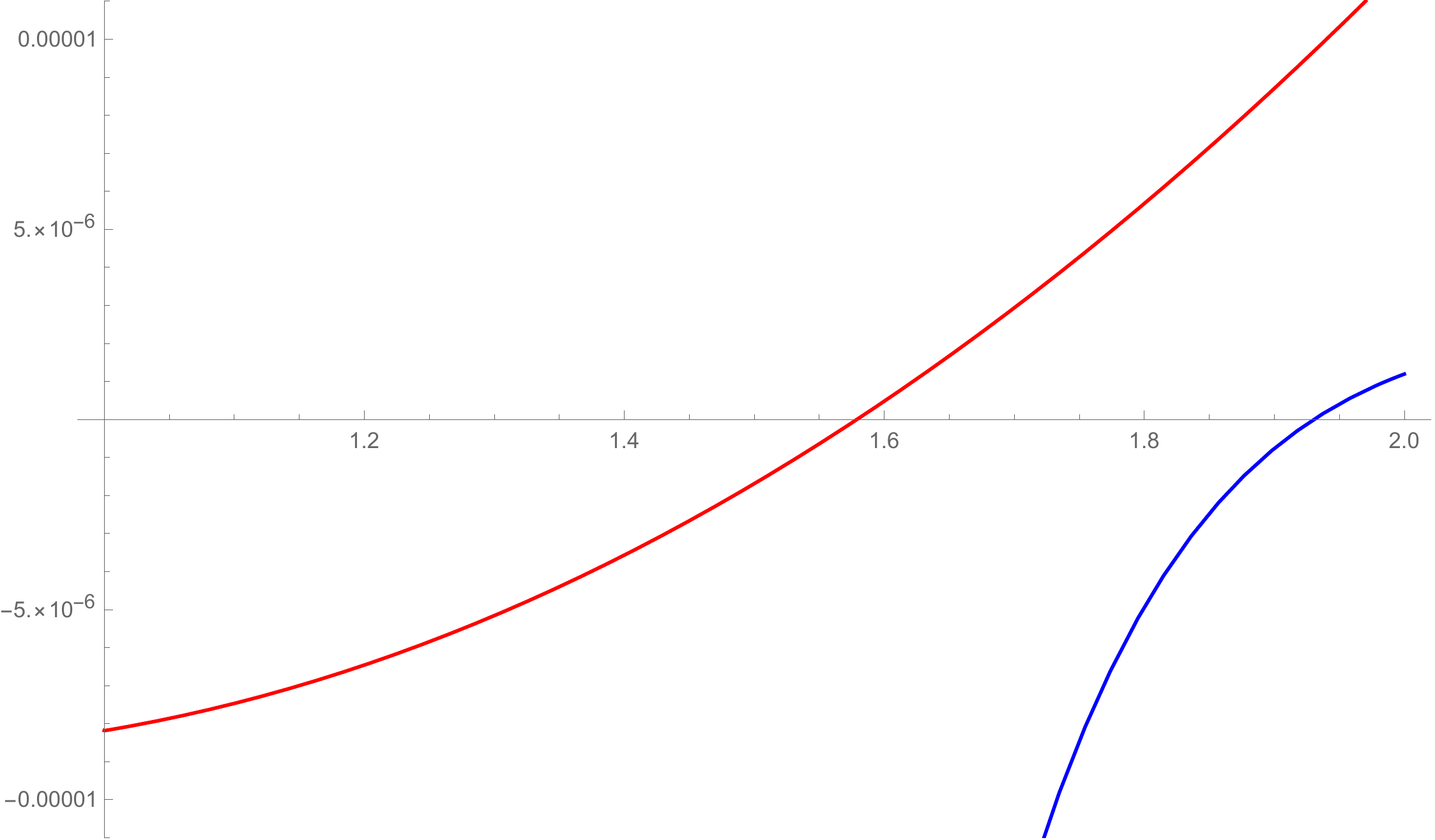}   
\put (-1,62) {$\Delta^A$}
\put (105,28) {$\alpha$}
\end{overpic}
\caption{\label{Fig:CE-A} 
The blue curve is $\Delta^A(BSC(10^{-3}), \alpha)$ and the red curve is $\Delta^A(BSC(0.46), \alpha)$. This shows that $\Delta^A$ can be positive or negative for the same value of $\alpha$, which means ultimately that $\kk^A$ cannot posses the desired convexity or concavity property.}
\end{figure}

 \begin{remark}\label{rem:CE-A}
The values in Lemma~\ref{CE-A} are those that we could conclusively verify. Furthermore, additional numerics suggests that the same methods should provide counterexamples in the range $\alpha \in (1.5783,2)$ when evaluated with sufficiently high precision.
\end{remark} 

Additionally, the numerical evidence leads us to making the following conjecture for the Arimoto conditional entropy.
\begin{conjecture}
There exists a value $1 < \hat\alpha < 1.5783$, such that 
 \begin{align}
 k_{\alpha}^A({k_\alpha^A}^{-1}(x)\ast {k_\alpha^A}^{-1}(y)) \label{kfunc3}
 \end{align}
 is convex for $0<\alpha<1$ and concave for $1<\alpha\leq\hat\alpha$. 
\end{conjecture}
Finally, remember that 
 \begin{align}
 \lim_{\alpha\rightarrow\infty} H^A_\alpha(X|Y) = H_\infty(X|Y). 
 \end{align}
 This value of $\alpha$ shows a particular behavior and deserves special attention. 
\begin{lem}
We have that $\kk^A_\infty(x,y)$ is independently linear in $x$ and $y$ and therefore 
\begin{align}
 H_\infty^A(X_1+X_2|Y_1Y_2) = h_\infty(h_\infty^{-1}( H_\infty^A(X_1|Y_1))\ast h_\infty^{-1}( H_\infty^A(X_2|Y_2))).
\end{align}
\end{lem} 
\begin{proof}
Remember from the definition, for binary $X$
 \begin{align}
k^A_\infty(p) = \max\{ p, 1-p \}. 
 \end{align}
Take $x\in[\frac12,1]$ and fix $c\geq\frac12$. Now, ${k^A_\infty}^{-1}(x)$ can be either $x\geq\frac12$ or $1-x\leq\frac12$; in the former case $x\ast c \leq\frac12$ and therefore $k^A_\infty(x\ast c)=1-x\ast c$, in the latter case $(1-x)\ast c \geq\frac12$ and $k^A_\infty((1-x)\ast c)=(1-x)\ast c = 1-x\ast c$. Therefore, the function  is identical and linear in both cases. Similarly one can show that when $c\leq\frac12$ then the function simply becomes $x\ast c$ which is again linear. 
\end{proof}
The last lemma is remarkable as it reproduces an equality that we usually only have for unconditioned entropies. 


\subsection{Bounds for $H_{\alpha}^{H}(X_1+X_2|Y_1Y_2)$} 

Similar to $H_{\alpha}^{A}$ we get the following equality for $H_{\alpha}^{H}$: 
 \begin{align} 
  H_{\alpha}^{H}(X|Y) &=\frac{1}{1-\alpha}\log\left(\sum_y p(y) e^{(1-\alpha) H_{\alpha}^{H}(X|Y=y) }\right) .
 \end{align}
 This motivates us to define the following quantity, 
 \begin{align}
  K_{\alpha}^{H}(X|Y) &= e^{(1-\alpha) H_{\alpha}^{H}(X|Y) } ,
 \end{align}
 which takes values on $[1, \delta_\alpha^H := 2^{1-\alpha} ]$. 
 For a more convenient notation we will furthermore use, in analogy to the binary entropy, $k_{\alpha}^{H}(p)$ whenever $X$ is a binary random variable with probability distribution $\{p,1-p\}$ (and trivial conditioning system): $K^H_\alpha(X) = k^H_\alpha(p)$.
 
   We will see that the crucial quantity here is the following
 \begin{align}
\kk^H_\alpha(x,y) = {k^H_{\alpha}}\left({k^H_\alpha}^{-1}(x)\ast {k^H_\alpha}^{-1}(y) \right).  \label{kfuncH}
 \end{align}
 
  Following the proof technique of the standard bounds on information combining, we get the following results. 
 \begin{thm}[$H_{\alpha}^{H}$ BSC-bound]
 If, for a given $\alpha$, the function $\kk^H_\alpha(x,y)$  is convex in $x$ for fixed $y$ and vice versa, then the following holds: \\
 If $\alpha >1$,
 \begin{align}
 H_\alpha^H(X_1+X_2|Y_1Y_2) \leq h_\alpha(h_\alpha^{-1}( H_\alpha^H(X_1|Y_1))\ast h_\alpha^{-1}( H_\alpha^H(X_2|Y_2))). \label{BSCupH}
 \end{align}
  If $\alpha <1$,
 \begin{align}
 H_\alpha^H(X_1+X_2|Y_1Y_2) \geq h_\alpha(h_\alpha^{-1}( H_\alpha^H(X_1|Y_1))\ast h_\alpha^{-1}( H_\alpha^H(X_2|Y_2))). \label{BSCdownH}
 \end{align}
If $\kk^H_\alpha(x,y)$ is concave instead, the inequalities hold with $\leq$ and $\geq$ exchanged. These bounds are optimal, in the sense that equality is achieved by binary symmetric channels. 
 \end{thm}
 \begin{proof}
The proof works analogous to that of Theorem~\ref{thm:BSC-A}. 
\end{proof}

\begin{thm}[$H_{\alpha}^{H}$ BEC-bound]\label{thm:BEC-H}
 If, for a given $\alpha$, the function $\kk^H_\alpha(x,y)$
 is convex in $x$ for fixed $y$ and vice versa, then the following holds: \\
 If $\alpha >1$, 
 \begin{align}
 H_\alpha^H(X_1+X_2|Y_1Y_2) &\geq \frac{1}{1-\alpha}\log\frac{(\delta_\alpha^H - K_\alpha^H(X_1|Y_1))(\delta_\alpha^H - K_\alpha^H(X_2|Y_2))}{1-\delta_\alpha^H} + \delta_\alpha^H  .
 \end{align}
  If $\alpha <1$, 
\begin{align}
 H_\alpha^H(X_1+X_2|Y_1Y_2) &\leq \frac{1}{1-\alpha}\log\frac{(\delta_\alpha^H - K_\alpha^H(X_1|Y_1))(\delta_\alpha^H - K_\alpha^H(X_2|Y_2))}{1-\delta_\alpha^H} + \delta_\alpha^H  .
 \end{align}
If $\kk^H_\alpha(x,y)$ is concave instead, the inequalities hold with $\leq$ and $\geq$ exchanged. These bounds are optimal, in the sense that equality is achieved by binary erasure channels. 
\end{thm}
 \begin{proof}
The proof works analogous to that of Theorem~\ref{thm:BEC-A}. 
\end{proof}

 Let's now consider the case of $\kk^{H}_\alpha$. 
\begin{lem}\label{lem:kdown}
The function
 \begin{align}
 k_{\alpha}^H({k_\alpha^H}^{-1}(x)\ast {k_\alpha^H}^{-1}(y)) \label{kfunc3}
 \end{align}
 is convex for $0<\alpha<1$ and $2 < \alpha\leq 3$ and concave for $1<\alpha\leq 2$ and $\alpha\geq 3$. 
\end{lem}
\begin{proof}
The proof is part of Appendix~\ref{Appendix1}. 
\end{proof}
Although technically included in the previous lemma, a particular case deserves some special attention. 
\begin{lem}\label{lem:kH-linear}
The function
 \begin{align}
 k_{\alpha}^H({k_\alpha^H}^{-1}(x)\ast {k_\alpha^H}^{-1}(y)) \label{kfunc3}
 \end{align}
 is linear in $x$ and $y$ for $\alpha=2$ and $\alpha = 3$. 
\end{lem}
\begin{proof}
Checked by explicit calculation, see Appendix~\ref{Appendix1}. 
\end{proof}
This lemma is remarkable as it tells us that the BSC-bound and the BEC-bound both hold with equality. We have
\begin{align}
 H_2^H(X_1+X_2|Y_1Y_2) &= h_2(h_2^{-1}( H_2^H(X_1|Y_1))\ast h_2^{-1}( H_2^H(X_2|Y_2))) \\[4pt]
  H_3^H(X_1+X_2|Y_1Y_2) &= h_3(h_3^{-1}( H_3^H(X_1|Y_1))\ast h_3^{-1}( H_3^H(X_2|Y_2))). 
\end{align}
These equations are noteworthy as they give an equality in the conditional case, something we usually only get for unconditioned entropies. 

To illustrate the results, we define in analogy to the technique used in Lemma~\ref{CE-A} the quantity
\begin{align}
\Delta^H(BSC(p),\alpha) := h(p\ast p)  -  \frac{1}{1-\alpha}\log\frac{(\delta_\alpha^H - K_\alpha^H(BSC(p)))^2}{1-\delta_\alpha^H} + \delta_\alpha^H,
\end{align}
which gives the difference between the BSC and the BEC bound evaluated on a binary symmetric channel with crossover probability $p$ and we plot it in Figure~\ref{Fig:E-H}. 

\begin{figure}[t!]
\centering
\begin{overpic}[scale=0.45]{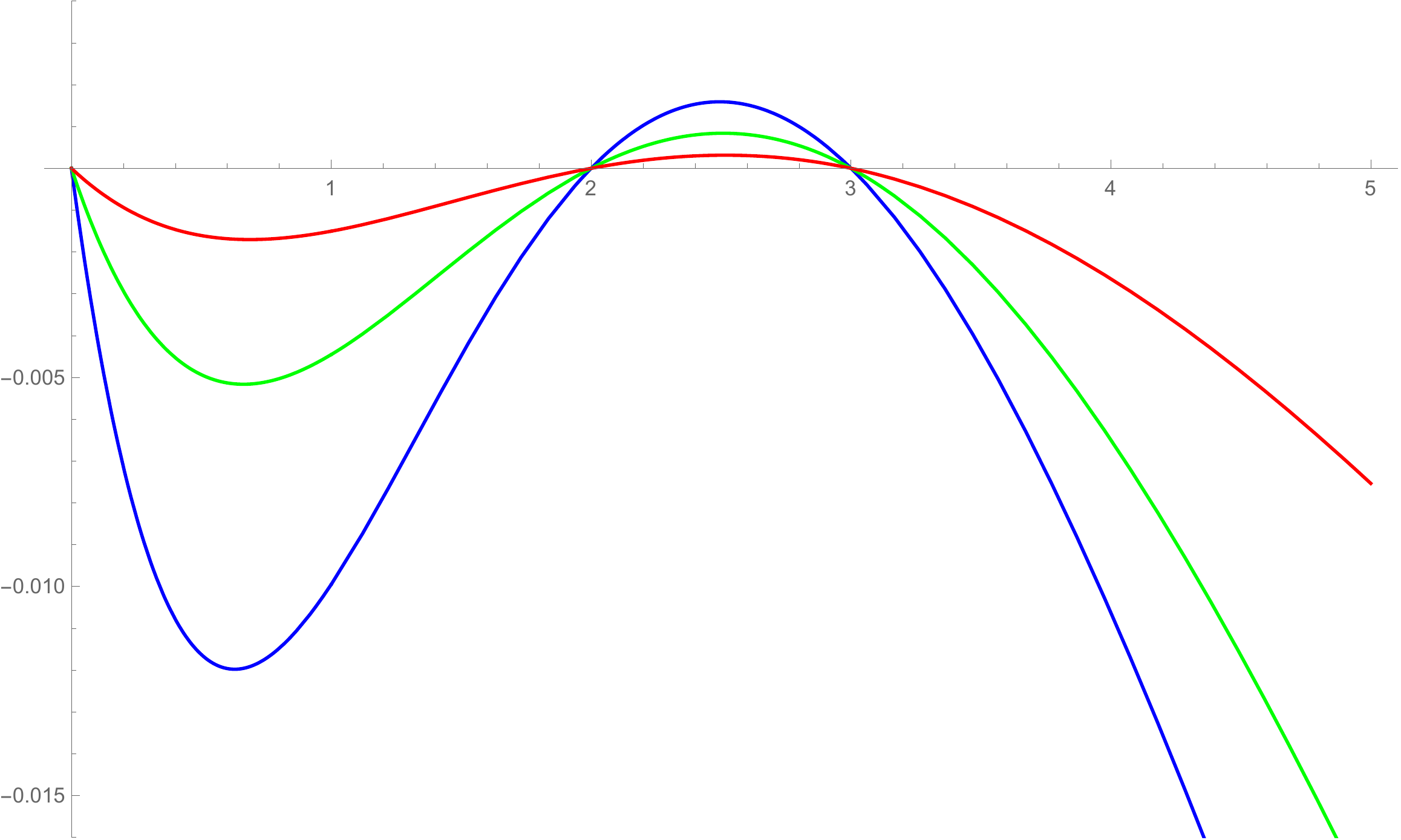}   
\put (-1,62) {$\Delta^H$}
\put (105,47) {$\alpha$}
\end{overpic}
\caption{\label{Fig:E-H} 
The different curves show $\Delta^H(BSC(p),\alpha)$ for different values of $p$. Blue is $p=0.25$,  green is $p=0.3$ and red is $p=0.35$. One can clearly observe how the BSC and the BEC bound switch their respective roles as upper and lower bound and that both are actually identical for $\alpha=2$ and $\alpha=3$. }
\end{figure}


\subsection{Bounds for $H_{\alpha}^{J}(X_1+X_2|Y_1Y_2)$} 

From the definition, we have
\begin{align}
H_\alpha^J(X|Y) &= \frac{1}{1-\alpha}\left[\log\sum_y p(y)^\alpha\sum_{x} p(x|y)^\alpha  - \log\sum_y p(y)^\alpha \right] \\
&= \frac{1}{1-\alpha}\left[\log\sum_y p(y)^\alpha e^{(1-\alpha) H_{\alpha}(X|Y=y) }  - \log\sum_y p(y)^\alpha \right]
\end{align}
This suggests that, again, the quantity 
 \begin{align}
  K_{\alpha}^{H}(X) &= e^{(1-\alpha) H_{\alpha}(X) } .
 \end{align}
 is of crucial interest here. However, due to the exponent $\alpha$ of $p(y)$ we can not directly invoke convexity, but have to play one extra trick compared to the proof of the previous bounds. 
 For the conditional version we will also need
  \begin{align}
  K_{\alpha}^{J}(X|Y) &= e^{(1-\alpha) H_{\alpha}^{J}(X|Y) } .
 \end{align}
 
 \begin{thm}[$H_{\alpha}^{J}$ BSC-bound]\label{thm:BSC-J}
If, for a given $\alpha$, the function $\kk^H_\alpha(x,y)$  is convex in $x$ for fixed $y$ and vice versa, then the following holds: \\
 If $\alpha >1$,
 \begin{align}
 H_\alpha^J(X_1+X_2|Y_1Y_2) \leq h_\alpha(h_\alpha^{-1}( H_\alpha^J(X_1|Y_1))\ast h_\alpha^{-1}( H_\alpha^J(X_2|Y_2))). \label{BSCupJ}
 \end{align}
  If $\alpha <1$,
 \begin{align}
 H_\alpha^J(X_1+X_2|Y_1Y_2) \geq h_\alpha(h_\alpha^{-1}( H_\alpha^J(X_1|Y_1))\ast h_\alpha^{-1}( H_\alpha^J(X_2|Y_2))). \label{BSCdownJ}
 \end{align}
If $\kk^H_\alpha(x,y)$ is concave instead, the inequalities hold with $\leq$ and $\geq$ exchanged. These bounds are optimal, in the sense that equality is achieved by binary symmetric channels. 
 \end{thm}
 \begin{proof}
To prove the theorem, we will show Equation \ref{BSCupJ} in the case of $\kk^A_\alpha(x,y)$ being convex. All the other combinations, for \ref{BSCupJ} and \ref{BSCdownJ}, follow directly. The prove strategy is similar to the proof in the Shannon entropy setting. 

To make use of the convexity property, define the following probability distribution (also sometimes called the \textit{tilted probability distribution}):
\begin{align}
\tilde p(y) := \frac{p(y)^\alpha}{\sum_y p(y)^\alpha}. \label{tildepy}
\end{align}

Consider the following chain of equations: 
 \begin{align} 
 H_{\alpha}^{J}(X_1+X_2|Y_1Y_2) &= \frac{1}{1-\alpha}\left[\log\sum_{y_1,y_2} p(y_1)^\alpha p(y_2)^\alpha e^{(1-\alpha) H_{\alpha}(X_1+X_2|Y_1=y_1, Y_2=y_2) }  - \log\sum_{y_1,y_2} p(y_1)^\alpha p(y_2)^\alpha \right] \\
 &= \frac{1}{1-\alpha}\left[\log\sum_{y_1,y_2} \tilde p(y_1) \tilde p(y_2) e^{(1-\alpha) H_{\alpha}(X_1+X_2|Y_1=y_1, Y_2=y_2) } \right] \\
	&=\frac{\alpha}{1-\alpha}\log\left(\sum_{y_1,y_2} \tilde p(y_1)\tilde p(y_2) {k^H_{\alpha}}\left({k^H_\alpha}^{-1}(K_\alpha^H(X_1|Y_1=y_1))\ast {k^H_\alpha}^{-1}(K_\alpha^H(X_2|Y_2=y_2)) \right) \right)	\\
	&\leq\frac{\alpha}{1-\alpha}\log\left({k^H_{\alpha}}\left({k^H_\alpha}^{-1}\left(\sum_{y_1} \tilde p(y_1)K_\alpha^H(X_1|Y_1=y_1)\right)\ast {k^H_\alpha}^{-1}\left(\sum_{y_2} \tilde p(y_2)K_\alpha^H(X_2|Y_2=y_2)\right) \right) \right)	 \\
	&= h_{\alpha}(h_\alpha^{-1}(H_\alpha^J(X_1|Y_1))\ast h_\alpha^{-1}(H_\alpha^J(X_2|Y_2))),
\end{align}
where all equalities follow simply by definition and rearranging. The inequality follows by using the convexity of $\kk^H_\alpha(x,y)$ twice, once in the first argument and once in the second. 
\end{proof} 

\begin{thm}[$H_{\alpha}^{J}$ BEC-bound]\label{thm:BEC-J}
 If, for a given $\alpha$, the function $\kk^H_\alpha(x,y)$
 is convex in $x$ for fixed $y$ and vice versa, then the following holds: \\
 If $\alpha >1$, 
 \begin{align}
 H_\alpha^J(X_1+X_2|Y_1Y_2) &\geq \frac{1}{1-\alpha}\log\frac{(\delta_\alpha^H - K_\alpha^J(X_1|Y_1))(\delta_\alpha^H - K_\alpha^J(X_2|Y_2))}{1-\delta_\alpha^H} + \delta_\alpha^H  .
 \end{align}
  If $\alpha <1$,
\begin{align}
 H_\alpha^J(X_1+X_2|Y_1Y_2) &\leq \frac{1}{1-\alpha}\log\frac{(\delta_\alpha^H - K_\alpha^J(X_1|Y_1))(\delta_\alpha^H - K_\alpha^J(X_2|Y_2))}{1-\delta_\alpha^H} + \delta_\alpha^H  .
 \end{align}
If $\kk^H_\alpha(x,y)$ is concave instead, the inequalities hold with $\leq$ and $\geq$ exchanged. These bounds are optimal, in the sense that equality is achieved by binary erasure channels. 
\end{thm}
 \begin{proof}
The proof works analogous to that of Theorem~\ref{thm:BEC-A} using the tilted probability distribution from the proof of Theorem~\ref{thm:BSC-J}.  
\end{proof}

Now the exact behavior of the bounds follows, as it did for $H_{\alpha}^{H}(X_1+X_2|Y_1Y_2)$, from Lemma~\ref{lem:kdown}. 
   
\subsection{Bounds for $H_{\alpha}^{C}(X_1+X_2|Y_1Y_2)$} 

Finally, we will investigate the properties of $H_{\alpha}^{C}$. On first glance, this might appear simpler then the other quantities since the sum over $y$ is outside the logarithm, which means that
\begin{align}
H_\alpha^C(X|Y) &= \sum_y p(y) H_\alpha(X|Y=y)
\end{align}
 and therefore the quantity to investigate will be
 \begin{align}
 \hh_\alpha(x,y) = {h_{\alpha}}\left({h_\alpha}^{-1}(x) \ast {h_\alpha}^{-1}(y) \right).   \label{hfunc}
 \end{align}
 However, it will turn out that proving convexity or concavity of this quantity is actually more complicated. Nevertheless, we will provide some partial results later in this section. First, let us state the BSC and BEC bound for this conditional entropy. 

  \begin{thm}[$H_{\alpha}^{C}$ BSC-bound]\label{thm:BSC-C}
If, for a given $\alpha$, the function $\hh_\alpha(x,y)$  is convex in $x$ for fixed $y$ and vice versa, then the following holds: \\
 If $\alpha >1$, 
 \begin{align}
 H_\alpha^C(X_1+X_2|Y_1Y_2) \leq h_\alpha(h_\alpha^{-1}( H_\alpha^C(X_1|Y_1))\ast h_\alpha^{-1}( H_\alpha^C(X_2|Y_2))). \label{BSCupC}
 \end{align}
  If $\alpha <1$, 
 \begin{align}
 H_\alpha^C(X_1+X_2|Y_1Y_2) \geq h_\alpha(h_\alpha^{-1}( H_\alpha^C(X_1|Y_1))\ast h_\alpha^{-1}( H_\alpha^C(X_2|Y_2))). \label{BSCdownC}
 \end{align}
If $\hh_\alpha(x,y)$ is concave instead, the inequalities hold with $\leq$ and $\geq$ exchanged. These bounds are optimal, in the sense that equality is achieved by binary symmetric channels. 
 \end{thm}
 \begin{proof}
The proof works analogous to that of Theorem~\ref{thm:BSC-A}. 
\end{proof}
 \begin{thm}[$H_{\alpha}^{C}$ BEC-bound]\label{thm:BEC-C}
 If, for a given $\alpha$, the function $\hh_\alpha(x,y)$
 is convex in $x$ for fixed $y$ and vice versa, then the following holds: \\
 If $\alpha >1$, 
 \begin{align}
 H_\alpha^C(X_1+X_2|Y_1Y_2) &\geq \log2 - \frac{(\log2 - H_\alpha^C(X_1|Y_1))(\log2 - H_\alpha^C(X_2|Y_2))}{\log2}.
 \end{align}
  If $\alpha <1$, 
\begin{align}
 H_\alpha^C(X_1+X_2|Y_1Y_2) &\leq \log2 - \frac{(\log2 - H_\alpha^C(X_1|Y_1))(\log2 - H_\alpha^C(X_2|Y_2))}{\log2}.
 \end{align}
If $\hh_\alpha(x,y)$ is concave instead, the inequalities hold with $\leq$ and $\geq$ exchanged. These bounds are optimal, in the sense that equality is achieved by binary erasure channels. 
\end{thm}
\begin{proof}
The proof works analogous to that of Theorem~\ref{thm:BEC-A}.  
\end{proof}

 \begin{lem}[Counterexample for $\hh_\alpha(x,y)$]\label{CE-C}
 For $\alpha\in(1.387, 1.95)$, the function
  \begin{align}
 h_{\alpha}\left({h_\alpha}^{-1}(x)\ast {h_\alpha}^{-1}(y)\right)
 \end{align}
 is neither convex nor concave in $x$ for fixed $y$ and vice versa. 
 \end{lem} 
 \begin{proof}
We will use the same strategy as for the counterexamples found in Lemma~\ref{CE-A}

For a binary symmetric channel with crossover probability $p$, denoted $BSC(p)$, we have 
 \begin{align}
 H_{\alpha}^{C}(BSC(p)) = h_\alpha(p)
 \end{align}
 and therefore the BSC bound for two of these channels is given by $h_\alpha(p\ast p)$. Evaluating the BEC bound for two $BSC(p)$ and defining the difference between both as
 \begin{align}
 \Delta^C(BSC(p),\alpha) := h(p\ast p)  -  \log2 + \frac{(\log2 - h_\alpha(p))^2}{\log2}
 \end{align}
 gives us an easily testable quantity. 
 We simply verify numerically that 
 \begin{align}
 \Delta^C(BSC(10^{-7}), \alpha) < 0 \quad\text{for}\; \alpha\in(1,1.95), \\
 \Delta^C(BSC(0.49), \alpha) > 0 \quad\text{for}\; \alpha\in(1.387,2).
 \end{align}
 By the above argument this leads to the desired result. 
\end{proof}
  \begin{remark}\label{rem:CE-C}
Similar to remark~\ref{rem:CE-A}, also here, numerics suggests that the same methods could provide counterexamples in the range $\alpha \in (1.3863,2)$ when evaluated with sufficiently high computer precision.
\end{remark} 
 We end this section with the following conjecture based on numerical evidence. 
 \begin{conjecture}
There exists a value $1 < \hat\alpha < 1.3863$, such that 
 \begin{align}
 h_{\alpha}\left({h_\alpha}^{-1}(x)\ast {h_\alpha}^{-1}(y)\right)
 \end{align}
 is convex for $0<\alpha<\hat\alpha$ and concave for $\alpha\geq 2$. 
\end{conjecture}

\begin{figure}
\centering
\begin{tikzpicture}
\filldraw[fill=red!30!white, draw=black] (0,0) rectangle (3,0.1);
\filldraw[fill=blue!30!white, draw=black] (3,0) rectangle (3*1.5783,0.1);
\filldraw[fill=yellow, draw=black] (3*1.5783,0) rectangle (6,0.1);
\filldraw[fill=red, draw=black] (6,0) rectangle (16,0.1);
\draw (0,-0.1) -- (0,0); 
\node at (0,-0.3){$0$};
\draw (3,-0.1) -- (3,0); 
\node at (3,-0.3){$1$};
\draw (6,-0.1) -- (6,0); 
\node at (6,-0.3){$2$};
\draw (9,-0.1) -- (9,0); 
\node at (9,-0.3){$3$};
\node at (16,-0.3){$\rightarrow\infty$};
\node at (-0.8,0.05){$\kk_\alpha^A$};
\draw[thick] (16,-1+1.1) -- (16,-0.75+1.1); 
\filldraw[fill=green, draw=black] (16,-0.88+1.1) -- (16,-0.75+1.1) -- (16.24,-0.75+1.1) -- (16.17,-0.815+1.1) -- (16.24,-0.88+1.1) -- cycle;

\filldraw[fill=red, draw=black] (0,-1) rectangle (3,-1.1);
\filldraw[fill=blue, draw=black] (3,-1) rectangle (6,-1.1);
\filldraw[fill=red, draw=black] (6,-1) rectangle (9,-1.1);
\filldraw[fill=blue, draw=black] (9,-1) rectangle (16,-1.1);
\draw (0,-1.2) -- (0,-1.1); 
\node at (0,-1.4){$0$};
\draw (3,-1.2) -- (3,-1.1); 
\node at (3,-1.4){$1$};
\draw (6,-1.2) -- (6,-1.1); 
\node at (6,-1.4){$2$};
\draw (9,-1.2) -- (9,-1.1); 
\node at (9,-1.4){$3$};
\node at (16,-1.4){$\rightarrow\infty$};
\node at (-0.8,-1.15){$\kk_\alpha^H$};
\draw[thick] (6,-1) -- (6,-0.75); 
\filldraw[fill=green, draw=black] (6,-0.88) -- (6,-0.75) -- (6.24,-0.75) -- (6.17,-0.815) -- (6.24,-0.88) -- cycle;
\draw[thick] (9,-1) -- (9,-0.75); 
\filldraw[fill=green, draw=black] (9,-0.88) -- (9,-0.75) -- (9.24,-0.75) -- (9.17,-0.815) -- (9.24,-0.88) -- cycle;

\filldraw[fill=red!30!white, draw=black] (0,-2) rectangle (3*1.3863,-2.1);
\filldraw[fill=yellow, draw=black] (3*1.3863,-2) rectangle (6,-2.1);
\filldraw[fill=blue!30!white, draw=black] (6,-2) rectangle (16,-2.1);
\draw (0,-2.2) -- (0,-2.1); 
\node at (0,-2.4){$0$};
\draw (3,-2.2) -- (3,-2.1); 
\node at (3,-2.4){$1$};
\draw (6,-2.2) -- (6,-2.1); 
\node at (6,-2.4){$2$};
\draw (9,-2.2) -- (9,-2.1); 
\node at (9,-2.4){$3$};
\node at (16,-2.4){$\rightarrow\infty$};
\node at (-0.8,-2.05){$\hh_\alpha$};

\filldraw[fill=red, draw=black] (0,-3.1) rectangle (1,-3.2);
\node at (1.75,-3.15){convex};
\filldraw[fill=red!30!white, draw=black] (2.5,-3.1) rectangle (3.5,-3.2);
\node at (4.75,-3.15){(conj.) convex};
\filldraw[fill=yellow, draw=black] (6,-3.1) rectangle (7,-3.2);
\node at (7.75,-3.15){neither};
\filldraw[fill=blue!30!white, draw=black] (8.5,-3.1) rectangle (9.5,-3.2);
\node at (10.75,-3.15){(conj.) concave};
\filldraw[fill=blue, draw=black] (12,-3.1) rectangle (13,-3.2);
\node at (13.75,-3.15){concave};
\draw[thick] (14.6,-3.25) -- (14.6,-2.95); 
\filldraw[fill=green, draw=black] (14.6,-2.2-0.88) -- (14.6,-2.2-0.75) -- (14.6+0.24,-2.2-0.75) -- (14.6+0.17,-2.2-0.815) -- (14.6+0.24,-2.2-0.88) -- cycle;
\node at (15.5,-3.15){linear};
\end{tikzpicture}
\caption{\label{Fig:Results} 
Summary of, proven or conjectured, convexity and concavity properties of the three considered functions $\kk_\alpha^A$, $\kk_\alpha^H$ and $\hh_\alpha$ for different values of $\alpha$. For precise statements we refer to the main text. }
\end{figure}
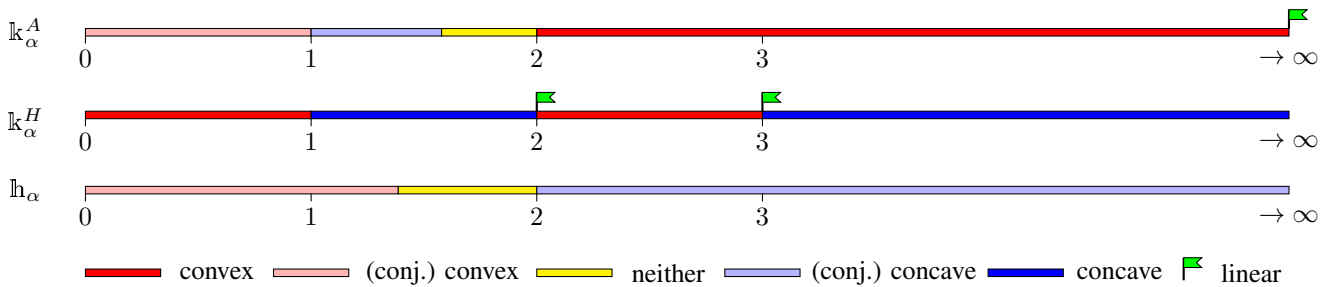

\section{Polarization under $H_{\alpha}^{J}$} 

Polar codes have gained a lot of attention recently as the first efficiently implementable codes that provably reach capacity. The underlying technique that enables Polar codes is that of polarization, allowing to transmit information with a rate equal to the capacity and vanishing error in the theoretical limit of infinite block length. In the original work by Arikan, polarization was proven for the Shannon entropy using a Martingale approach. Later, Alsan and Teletar~\cite{AT14} gave a conceptually simpler proof using information combining bounds, in particular the Mrs. Gerbers Lemma. Very recently, polarization was shown to also happen when considering the $H_{\alpha}^{J}$ R\'enyi entropy, again via a Martingale approach~\cite{zheng2019polarization}. As the main result of this section we show that our R\'enyi information combining bounds from the last section can be used to give a simpler and more intuitive proof following the technique established in~\cite{AT14}. 

Let $I(W)$ be the symmetric capacity of a channel as expressed by its (Shannon) mutual information. In its operational interpretation as the channel capacity, the mutual information is uniquely defined and it is known to exhibit polarization under Arikans channel transformation. However, it appears natural to ask whether polarization also occurs when considering a different quantity.

Let $W^+$ and $W^-$ denote the channels resulting from transforming two copies of the channel $W$. By carefully going through the simpler polarization proof in~\cite{AT14} one notices that essentially only two properties of the symmetric capacity are needed: 
\begin{align}
I(W^+) + I(W^-) &= 2 I(W) \label{chainrule}\\
\Delta(W) := \frac12 [ I(W^+) - I(W^-) ] &\geq \kappa(a,b), \label{growth}
\end{align}
where for the inequality $I(W)\in [a,b]$ and $\kappa(a,b)>0$ whenever $0<a<b<1$. 

Now lets consider an alternative quantity $I^*(W)$, one can easily generalize the main result in~\cite{AT14} to be formulated in the following way.
\begin{thm}
Let $I^*(\cdot)$ be a channel mutual information that fulfills the conditions in Equations~\ref{chainrule} and~\ref{growth}. For any binary input channel $W$, and any $0<a<b<1$, we have
\begin{align}
\lim_{n\rightarrow\infty} \frac{1}{2^n} \#\{ s^n\in \{+,-\}^n : I^*(W^{s^n})\in [0,a) \} &= 1 - I^*(W), \\
\lim_{n\rightarrow\infty} \frac{1}{2^n} \#\{ s^n\in \{+,-\}^n : I^*(W^{s^n})\in [a,b] \} &= 0, \\
\lim_{n\rightarrow\infty} \frac{1}{2^n} \#\{ s^n\in \{+,-\}^n : I^*(W^{s^n})\in (b,1] \} &= I^*(W), 
\end{align}
\end{thm}
\begin{proof}
As described before, the proof follows immediately by following the proof of~\cite[Theorem 1]{AT14}.
\end{proof}
We are just left with showing that the above can be applied to a mutual information based on $H_{\alpha}^{J}$. Let us take $X$ and $Y$ as the random variables describing the input and the output of the channel $W$, respectively, and define 
\begin{align}
I_{\alpha}^{J}(W) :&= H_\alpha(X) - H_{\alpha}^{J}(X|Y) \\
&=  \frac{1}{1-\alpha}\left[\log\sum_{x,y} p(x,y)^\alpha - \log\sum_x p(x)^\alpha  - \log\sum_y p(y)^\alpha \right]. \nonumber
\end{align}
This is just one possible definition of an $\alpha$-mutual information but it turns out to be very convenient for our purpose. For a more general discussion on other definitions we refer to~\cite{verdu2015alpha}. 

Now, its easy to see that the definition above obeys the chain rule the same way $H_{\alpha}^{J}$ does and therefore a standard proof applies to show that $I_{\alpha}^{J}(W)$ obeys the condition in Equation~\ref{chainrule}. Additionally, it is rather easy to see that our results in the previous section ensure that $I_{\alpha}^{J}(W)$ also obeys the condition in Equation~\ref{growth} (compare e.g.~\cite[Lemma 1]{AT14}. Therefore it follows directly that $I_{\alpha}^{J}(W)$ exhibits polarization and, since $I_{\alpha}^{J}(W) = \log2 - H_{\alpha}^{J}(X|Y)$ for binary symmetric channels and flat $X$, so does $H_{\alpha}^{J}$, which is the main result of~\cite{zheng2019polarization}. 

We finish the section with two remarks. First, in~\cite{AT14} the authors also consider the case where the channel transformation is applied to non identical channels. We remark that following the same ideas as above the result can easily be generalized to hold for generalized $I^*(W)$ that fulfill the corresponding generalizations of Equations~\ref{chainrule} and~\ref{growth}. This holds in particular for $I_{\alpha}^{J}(W)$ and therefore also non-stationary channels polarize when considering this measure. 

Second, since in the Shannon information combining bounds the equality in the lower bound is achieved by the binary symmetric channel, the BSC seems to be the channel that is the most difficult to polarize. However, when considering $H_{\alpha}^{J}$, for some values of $\alpha$ the lower bound is satisfied by the binary erasure channel, suggesting that the roles of the channels are inverted when it comes to polarization.

%
%

\section{Conclusions} 
We have shown several extensions of information combining bounds to R\'enyi entropies. In general, the results require a convexity or concavity property, we have however been able to show such a property for several settings. Additionally, we discussed a simple application of our results to the polarization of R\'enyi entropies. 

Due to the number of possible definitions of conditional R\'enyi entropies we decided to focus on four of them. Primarily, they were chosen since all of them have found several applications in information theory and are therefore of practical relevance. Additionally, they all share the property that taking the limit $\alpha\rightarrow 1$ gives the Shannon entropy case. The latter can be useful as it gives an alternative approach towards proofing this special case. While this doesn't seem of particular interest in the traditional setting, recently the extension of information combining bounds to the realm of quantum information theory has been discussed~\cite{HR17}, where proving optimal bounds remains an open problem, even for the von Neumann entropy which is the natural generalization of the Shannon entropy to the quantum setting. Proving a conjectured lower bound would have immediate applications investigating Polar codes for classical-quantum channels~\cite{HR17, WG11, H14}. A potential strategy towards solving the quantum case, could be to look at quantum R\'enyi entropies and take the limit to the von Neumann entropy. In particular, $H_\alpha^A$ and $H_\alpha^H$ are special cases of the most commonly used quantum R\'enyi entropies~\cite{tomamichel2015quantum}. Our results can therefore be seen as partial progress towards a new approach to solve the quantum case. 

Finally, this work leaves some natural open problems. The most obvious one being to determine whether the desired convexity or concavity properties hold in those ranges of $\alpha$ where it is not yet known. Additionally, it remains open whether one can find different bounds that hold also where the convexity or concavity does not hold. One might also consider entropies other than the R\'enyi entropies. Particularly interesting ones would again be those that contain the Shannon entropy as a special case. Finally, considering different combining operations or non-binary $X$ would be an interesting goal for future research, compare e.g.~\cite{JA12, GV14, GB15, madiman2017entropy, madiman2019majorization}. 

\section*{Acknowledgments} 
 CH acknowledges financial support from the VILLUM FONDEN via the QMATH Centre of Excellence (Grant no. 10059).
 

\bibliographystyle{IEEEtran}
\bibliography{biblio}


\appendix
\subsection{Proof of Lemma~\ref{lem:kdown} and Lemma~\ref{lem:kH-linear}}\label{Appendix1}

Our goal is to investigate properties of the function
\begin{align}
\kk^H_{\alpha}(x) = k^H_{\alpha}({k_{\alpha}^H}^{-1}(x)\ast c),
\end{align}
in particular concavity, convexity or linearity in $x$. The results in Lemma~\ref{lem:kdown} and Lemma~\ref{lem:kH-linear} then follow by the symmetry of $\kk^H_\alpha(x,y)$ under exchanging $x$ and $y$. 

Recall that a function $f(x)$ is convex iff
\begin{align}
f(y) \geq f(x) + f'(x)(y-x) \label{convcond}
\end{align}
for all $x,y$. Concavity holds if above inequality holds as $\leq$ and the function is linear if it becomes an equality. 

It can easily be seen that 
\begin{align}
\partial_x \kk^H_{\alpha}(x) = \kk^H_{\alpha}(x) \frac{1-2c}{x}\frac{h'_\alpha({k^H_{\alpha}}^{-1}(x)\ast c)}{h'_\alpha({k^H_{\alpha}}^{-1}(x))}\,.
\end{align}
With some rewriting Equation~\eqref{convcond} becomes
\begin{align}
\frac{\kk^H_{\alpha}(y)}{\kk^H_{\alpha}(x)} - \frac{y}{x}(1-2c)\frac{h'_\alpha({k^H_{\alpha}}^{-1}(x)\ast c)}{h'_\alpha({k^H_{\alpha}}^{-1}(x))} \geq 1-(1-2c)\frac{h'_\alpha({k^H_{\alpha}}^{-1}(x)\ast c)}{h'_\alpha({k^H_{\alpha}}^{-1}(x))} \,.
\end{align}
Since this has to hold for all $x$ and $y$ in the parameter range we can do the re-parametrization $x\rightarrow k^H_{\alpha}(x)$ and $y\rightarrow k^H_{\alpha}(y)$, leading to
\begin{align}
\frac{k^H_{\alpha}(y\ast c)}{k^H_{\alpha}(x\ast c)} - \frac{k^H_{\alpha}(y)}{k^H_{\alpha}(x)}(1-2c)\frac{h'_\alpha(x\ast c)}{h'_\alpha(x)} \geq 1-(1-2c)\frac{h'_\alpha(x\ast c)}{h'_\alpha(x)} \,.
\end{align}
Let's define the function 
\begin{align}
g_{\alpha}^x(y) := \frac{k^H_{\alpha}(y\ast c)}{k^H_{\alpha}(x\ast c)} - \frac{k^H_{\alpha}(y)}{k^H_{\alpha}(x)}(1-2c)\frac{h'_\alpha(x\ast c)}{h'_\alpha(x)} \,.
\end{align}
Now it can be easily seen that the inequality reduces to
\begin{align}
g_{\alpha}^x(y) \geq g_{\alpha}^x(x) \,.
\end{align}
The original convexity/concavity problem is therefore reduced to whether the function $g_{\alpha}^x(y)$ has a minimum or maximum at $x$. 

To simplify the expression a bit more we consider the following simple calculation,
\begin{align}
h'_\alpha(x) = \frac{\alpha}{1-\alpha}\frac{x^{\alpha-1}-(1-x)^{\alpha-1}}{k_{\alpha}^H(x)}\,.
\end{align}
We find,
\begin{align}\label{Eqn:alphabeta}
g_{\alpha}^x(y) = \frac{1}{k_{\alpha}^H(x\ast c)} \left[ k_{\alpha}^H(y\ast c)  -  k_{\alpha}^H(y) (1-2c) \frac{(x\ast c)^{\alpha-1}-(1-x\ast c)^{\alpha-1}}{x^{\alpha-1}-(1-x)^{\alpha-1}}     \right]\,.
\end{align}
As a start, we consider two special cases that will turn out to have an interesting property: Consider $\alpha=2$ and $\alpha=3$. First observe that interestingly, 
\begin{align}
\frac{(x\ast c)^{\alpha-1}-(1-x\ast c)^{\alpha-1}}{x^{\alpha-1}-(1-x)^{\alpha-1}} \Bigr|_{\alpha=2} = \frac{(x\ast c)^{\alpha-1}-(1-x\ast c)^{\alpha-1}}{x^{\alpha-1}-(1-x)^{\alpha-1}} \Bigr|_{\alpha=3} = (1-2c)
\end{align}
and also
\begin{align}
k_{2}^H(y\ast c) &= (1-c)^2 + c^2 - 2y(1-y)(1-2c)^2 \\
k_{3}^H(y\ast c) &= 1-3c+3c^2 - 3y(1-y)(1-2c)^2 \,.
\end{align}
Therefore,
\begin{align}
g_{2,1}^x(y) &= \frac{1}{k_{2}^H(x\ast c)} \left[  (1-c)^2 + c^2   -  (1-2c)^2     \right]\,, \\
g_{3,1}^x(y) &= \frac{1}{k_{3}^H(x\ast c)} \left[   1-3c+3c^2  -  (1-2c)^2     \right]\,
\end{align}
which makes the function independent of $y$ in both cases. It follows that for these values of $\alpha$ and $\beta$ Equation~\eqref{convcond} holds even with equality and therefore $\kk$ is linear, which proves Lemma~\ref{lem:kH-linear}.  

Let's further investigate Equation~\ref{Eqn:alphabeta}, for which taking the derivative in $y$ gives:
\begin{align}
(g_{\alpha}^x)'(y) = \frac{\alpha (1-2c)}{k_{\alpha}^H(x\ast c)} \left[  \left[  y_c^{\alpha-1} - (1-y_c)^{\alpha-1}\right] - \frac{ x_c^{\alpha-1} - (1-x_c)^{\alpha-1}}{ x^{\alpha-1} - (1-x)^{\alpha-1}} \left[ y^{\alpha-1} - (1-y)^{\alpha-1} \right]  \right],
\end{align}
where we use $y_c := y\ast c$ and $x_c := x\ast c$. 

Obviously we have $(g_{\alpha}^x)'(x) = 0$ as needed. Further define $f(x):=x^{\alpha-1} - (1-x)^{\alpha-1}$ which gives
\begin{align}
(g_{\alpha}^x)'(y) = \frac{\alpha (1-2c)}{k_{\alpha}^H(x\ast c)} \left[  f(y_c) - \frac{ f(x_c)}{ f(x)} f(y) \right],
\end{align}

Note that $y_c \geq y$ for $0\leq y\leq 0.5$ and $y_c \leq y$ for $0.5\leq y\leq 1$. Since the prefactor is always positive and independent of $y$, Lemma~\ref{lem:kdown} follows from the properties of the function $f(x)$ and in particular its first and second derivative (note that these have factors $(\alpha-1)$ and $(\alpha-2)(\alpha-1)$ respectively and latter has exponent $\alpha-3$, which leads to the change between convexity and concavity at $\alpha=1,2,3$).  


\end{document}